\pdfoutput=1 
\documentclass[letterpaper, 10 pt, conference]{ieeeconf}  

\IEEEoverridecommandlockouts                              

\usepackage{amsmath} 
\usepackage{amssymb}  
\usepackage{cite}
\usepackage{amsfonts}
\usepackage{amsmath}
\usepackage{mathrsfs}
\usepackage{soul,color}
\usepackage{multirow}
\usepackage{adjustbox}
\usepackage[dvipsnames]{xcolor}
\usepackage{comment}
\usepackage[normalem]{ulem}
\usepackage{amssymb,graphicx}

\NewDocumentCommand{\tra}{om}{%
	\IfNoValueTF{#1}
	{#2}
	{#2_{[#1]}}%
}

\NewDocumentCommand{\pce}{om}{%
	\IfNoValueTF{#1}
	{\mathsf{#2}}
	{\mathsf{#2}^{#1}}%
}

\NewDocumentCommand{\basis}{om}{%
	\IfNoValueTF{#1}
	{\| \phi^{#2}\|^2}
	{\langle \phi^{#1},\phi^{#2} \rangle}%
}

\usepackage{algorithm}
\usepackage{algorithmicx}  
\usepackage{algpseudocode}
\usepackage{booktabs}

\usepackage{tikz}
\usetikzlibrary{positioning,shapes,arrows.meta}
\usetikzlibrary{backgrounds}

\definecolor{gray1}{rgb}{0.9,	0.91,	0.93}
\definecolor{lightblue}{rgb}{0.6, 0.76, 0.92}
\definecolor{gray2}{rgb}{0.95, 0.95,0.95	}
\definecolor{gray3}{rgb}{0.85, 0.85, 0.85	}
\definecolor{lightblue1}{rgb}{0.87, 0.92, 0.97	}
\definecolor{lightblue2}{rgb}{0.62, 0.76, 0.90	}
\definecolor{lightblue3}{rgb}{0.36	0.61	0.84}

\newcommand{\End}{\hfill $\square$}

\DeclareMathOperator*{\trace}{tr}
 
\newcommand{\mbb}[1]{\mathbb{#1 }}
\newcommand{\mbf}[1]{\mathbf{#1}} 
\newcommand{\mcl}[1]{\mathcal{#1}}

\newcommand{\pcecoe}[2]{\mathsf{#1}^{#2}}
\newcommand{\relx}{(\omega)}

\newcommand{\trar}[2]{#1_{[0,{#2}]}}
\newcommand{\set}[1]{\mathbb{#1}}    

\newcommand{\inst}[1]{_{#1}}

\newcommand{\splx}[1]{\mcl{L}^2(\Omega, \mathcal{F}, \mu; \mathbb{R}^{#1})}
\newcommand{\spl}{\mcl{L}^2(\Omega, \mathcal{F}, \mu; \mathbb{R})} 
 
\newcommand{\diff}{\mathop{}\!\mathrm{d}}

\newcommand{\mean}{\mbb{E}}

\newcommand{\prob}{\mbb{P}}
\newcommand{\hankel}{\mcl{H}}

\newcommand{\Tini}{{T_{\text{ini}}}}

\newcommand{\ini}{{\text{ini}}} 

\newcommand{\dimx}{{n_x}}
\newcommand{\dimy}{{n_y}}
\newcommand{\dimu}{{n_u}}
\newcommand{\dimw}{{n_w}}

\newcommand{\dimv}{{n_v}}

\newcommand{\I}{\mathbb{I}}
\newcommand{\N}{\mathbb{N}}
\newcommand{\R}{\mathbb{R}}

\newtheorem{theorem}{Theorem}

\newtheorem{lemma}{Lemma}
\newtheorem{proposition}{Proposition}
\newtheorem{definition}{Definition}

\newtheorem{assumption}{Assumption}

\newcommand{\wmean}{ m}
\newcommand{\wcov}{ \Gamma}

\usepackage{bbm}
\title{Distributionally robust uncertainty quantification \\via data-driven stochastic optimal control}

\author{Guanru Pan and Timm Faulwasser$^{\star}$
\thanks{
	$^{\star}$: Corresponding author. }%
\thanks{The authors gratefully acknowledge funding by the German Research Foundation (Deutsche Forschungs-
gemeinschaft DFG) under project number 499435839. GP and TF are with the Institute for Energy Systems, Energy Efficiency and Energy Economics, TU Dortmund University, Dortmund,  Germany
       {\tt\small guanru.pan@tu-dortmund.de and timm.faulwasser@ieee.org} 
}
}

\begin{document}

\maketitle
\thispagestyle{empty}
\pagestyle{empty}

\begin{abstract}
This paper studies  optimal control problems of unknown linear systems 
subject to  stochastic disturbances of uncertain distribution.  Uncertainty about the stochastic disturbances is usually described via ambiguity sets of probability measures or distributions. 
Typically, stochastic optimal control requires knowledge of underlying dynamics and is as such challenging.
 Relying on a stochastic fundamental lemma from data-driven control and on the framework of polynomial chaos expansions, we propose an approach to reformulate distributionally robust optimal 
control problems with  ambiguity sets as uncertain conic programs in a finite-dimensional vector space. We show how to construct these programs from previously recorded data and how to relax the uncertain 
conic program to numerically tractable convex programs via appropriate 
sampling of the underlying distributions. The efficacy of our method is 
illustrated via a numerical example.\vspace*{1mm}
\end{abstract}


\section{INTRODUCTION}
In many real-world applications, stochastic disturbances pose significant challenges, such as distributed energy systems facing uncertain wind speed and renewable energy generation, or building control systems dealing with uncertain weather conditions and occupancy.
To hedge against the uncertainty surrounding the disturbance statistics, distributionally robust formulations optimize over an \textit{ambiguity set} of possible disturbance distributions ensuring robust satisfaction of equality and inequality constraints~\cite{Wiesemann2014}.
Additionally, the complexity and time-consuming nature of first principles modeling and system identification further motivate the need for data-driven approaches.

	There are two prominent data-driven avenues to  distributionally robust optimal control:  data-based synthesis of ambiguity sets to capture the uncertainty surrounding the distribution of disturbances while requiring explicit knowledge of a system model \cite{Fochesato2022,Lu2020,Coppens2021} and  robustness analysis of data-driven system descriptions with respect to uncertainty surrounding the distribution of the measurement noise \cite{Coulson21a}. However,  uncertainty propagation through dynamics without explicit knowledge of the system model and considering distributional uncertainty of the disturbance is still an open problem. In this work, we address this gap by generalizing the data-driven description of stochastic linear systems  based on Polynomial Chaos Expansion (PCE) from~\cite{Pan2022a,Faulwasser2022} towards  uncertainty surrounding the disturbance distribution. 
	
	Specifically, the present paper appears to be the first to combine data-driven descriptions of stochastic systems via PCE and Hankel matrices, exact convex reformulation of Gelbrich ambiguity sets, and exact reformulation of chance constraints towards  distributionally robust stochastic optimal control without explicit model knowledge. The main contributions are threefold:
	(i) we present a novel formulation of ambiguity sets for distributionally robust optimization using PCE including  an exact convex reformulation for Gelbrich ambiguity sets.
	Moreover, while \cite{VanParys2015,Coppens2021,Coulson21a} use the conditional value at risk to reformulate chance constraints, we consider an exact reformulation applicable under distributional uncertainty. (ii) in contrast to \cite{Li2021}, which considers ambiguity sets specified by fixed values of the first two moments, we  allow for ranges of the first two moments via Gelbrich sets.
	(iii) we present mild conditions under which a distributionally robust Optimal Control Problem (OCP) with Gelbrich ambiguity and stated in random variables can be equivalently reformulated as an uncertain conic program without explicit knowledge of the system matrices. We also propose an approach to approximate this uncertain conic program with sampled uncertainty distributions. Finally, we draw upon a simulation example to demonstrate the efficacy of the proposed scheme.

\subsubsection*{Notation}
Given a vector $x \in \R^n$ and a matrix $M \in \R^{n\times m}$, we specify $\|x\|$ as the 2-norm and $\|M\| = \sqrt{\trace(MM^\top)}$ as the Frobenius norm.
We denote the set of all positive semi-definite (positive definite) matrices in $\R^{n \times n}$ as $\mbb{S}_+^{n}$ ($\mbb{S}_{++}^{n}$). The principal square root of  $Q \in \mbb{S}_+^{n} $ is  written as $Q^{\frac{1}{2}}$.
The vectorization of  $\{x_k\}_{k=0}^{N-1}$ is denoted as $x_{[0,N-1]}$.

\section{Problem statement and preliminaries} \label{sec:problem_statement}
We first revisit the essential notions of probability theory. For rigorous definitions, we refer to the textbook~\cite{sullivan15introduction}. A measurable space is a pair $(\Omega,\mcl F)$ where $\Omega$ is the sample space and $\mcl F$ is a $\sigma$-algebra on $\Omega$. A \textit{probability measure} on the measurable space $(\Omega,\mcl F)$ is a function $\mu : \mcl F \to [0,1]$ with $\mu(\Omega)=1$. The triple $(\Omega,\mcl F,\mu)$ is a probability space. A \textit{random variable} $V$ is a  measurable function $V : \Omega \to \R^{\dimv}$ from the probability space $(\Omega,\mcl F,\mu)$ to the  measurable space $(\R^{\dimv},\mcl{B})$ where $\mcl B$ represents the Borel $\sigma$-algebra. Moreover, an  $\mcl L^2$ random variable $V \in \splx{\dimv}$ is finite in the $\mcl L^2$ norm, i.e.,
$
\|V\|^2 \doteq \int_{\Omega}{V(\omega)^\top V(\omega)} \textrm{d}\mu(\omega)< +\infty.
$
The random variable $V$ induces the probability measure $\mu_V  $ on $(\R^{\dimv},\mcl{B})$, i.e., for all $\mcl E \in  \mcl{B}$,
$
\mu_V(\mcl E) = \mu(\{\omega \in \Omega \,|\,V(\omega) \in \mcl E\}),$
denoted as the \textit{distribution} of the random variable $V$. For compactness, we write $V \sim \mu_V$. 
Consider  two random variables $V,\tilde{V} \in \splx{\dimv}$. The expectation of $V$ is written as $\mean[V] \in \R^\dimv$, its variance is $\mbb{V}[V] \in \R^\dimv$, and the covariance of $V$ and $\tilde{V}$ is denoted by $\Sigma [V,\tilde{V}]$. 
	\begin{definition}[Gelbrich distance \cite{Givens1984}] \label{def:Gelbrich}
		 Consider two tuples of mean vectors and covariance matrices $(\wmean, \wcov)$ and $(\bar{\wmean}, \bar{\wcov})$, their  Gelbrich distance  $\mbb G ((\wmean, \wcov),(\bar{\wmean}, \bar{\wcov})) \doteq d \geq 0$ is given by
	\[
			d=\sqrt{\|\wmean-\bar{\wmean}\|^2 + \trace( \wcov + \bar{\wcov }- 2( \bar{\wcov }^{\frac{1}{2}}  \wcov \bar{\wcov }^{\frac{1}{2}} )^{\frac{1}{2}} )}.
	\]
	\end{definition}

\subsection{Stochastic linear time-invariant systems}
We consider stochastic discrete-time Linear Time-Invariant~(LTI) systems
\begin{subequations}\label{eq:RVdynamics}
	\begin{align}
		X\inst{k+1} &= AX\inst{k} +BU\inst{k}+ EW\inst{k},\quad X_0=x_\ini \\
		Y\inst{k} &= CX\inst{k} + DU\inst{k} +F  W\inst{k}, 
	\end{align}
\end{subequations}
with state $X\inst{k}\in \mcl L^2(\Omega, \mcl F,\mu;\R^{n_x})$, input $U\inst{k}\in \mcl L^2(\Omega, \mcl F,\mu;\R^{\dimu})$, output $Y\inst{k}\in \mcl L^2(\Omega, \mcl F,\mu;\R^{\dimy})$, and stochastic disturbance $W\inst{k}\in \mcl L^2(\Omega, \mcl F,\mu;\R^{\dimw})$ for $k\in \N$. 
Note that the stochastic processes $X$, $Y$, and $U$ are adapted to the  filtration containing all historical information, cf.~\cite{sullivan15introduction}.
In this paper, we consider a deterministic initial condition $x_\ini \in \R^{\dimx}$ for  \eqref{eq:RVdynamics} and identically independently distributed (i.i.d.) (not necessarily Gaussian) disturbances  $W_k\sim \mu_W$.

	Instead of exact knowledge of $\mu_W$,  we model it as an element of a given ambiguity set. The most commonly used ambiguity sets employ the Wasserstein metric.
	However, tractable reformulations of  Wasserstein ambiguity sets are limited to certain empirical distributions \cite{Wiesemann2014} or to  ambiguity sets comprising  Gaussians \cite{Nguyen2021}.
	As an alternative,  \textit{Gelbrich ambiguity sets} include  all  distributions with moments that closely match a given empirical  pair $(\bar{m},\bar{\wcov})$ 
	based on the Gelbrich distance in Definition~\ref{def:Gelbrich}. Specifically, we consider the  Gelbrich ambiguity set with a given radius $\rho \in \R^+$
	 \begin{equation}\label{eq:ambiguity_WS}
	 	\mcl A \doteq \left\{ \mu_W\in (\R^{\dimw},\mcl{B})\,\middle|\,\begin{array}{c}
	 		\mu_W \in \mcl{D}(m,\Gamma),\, \Gamma\succeq 0
	 		\\  \mbb G \left((m, \Gamma), (\bar{m},\bar{\wcov}) \right)\leq \rho
	 	\end{array}\right\} .
	 \end{equation}
	Here $\mcl D(\wmean, \wcov)$ is  the set of distributions with mean $\wmean \in \R^{\dimw}$ and covariance $\wcov \in \mbb{S}_+^{\dimw}$ . 
	It is worth to be noted that  the Gelbrich ambiguity set is an outer approximation for the corresponding Wasserstein set~\cite{Givens1984}.
	Additionally, we remark  that \cite{Li2021}  considers the special case of more restrictive ambiguity sets with fixed first two moments, i.e. $\mcl D(\bar m, \bar \Gamma)$. This corresponds to Gelbrich  sets with $\rho = 0$.

Moving from distributions (or probability measures) to random variables, we note that the ambiguity set induces an  uncertainty set for the sequence of random variables $W_{[0,N-1]}$  with respect to  $N\in\N$
\begin{equation}\label{eq:uncertainty_RV}
	\mcl W\doteq \left\{	W_{[0,N-1]}\, \middle|\,
	\begin{gathered}
	 \forall i \neq k, \quad i,k \in [0, N-1],  \\
	 W_k \sim \mu_W\in \mcl A,\,\Sigma[W_k,W_i] = 0
	\end{gathered}  \right\}.
\end{equation}
Note that $\mcl A \subset (\R^{\dimw},\mcl{B}) $ while $\mcl W  \subset  (\splx{\dimw})^N$.

\subsection{Model-based distributionally robust  optimal control}
Our analysis begins with a distributionally robust OCP with the explicit knowledge of the system model,
while   its data-driven counterpart is presented in Section~\ref{sec:data-drivenOCP}. 
Consider the uncertainty set \eqref{eq:uncertainty_RV}, we have 
\begin{subequations} \label{eq:stochasticOCP}
		\begin{align}
				\min_{ 	\substack{
					\bar{u}, \,K,\, \alpha, 
						{U},\,  {Y}, \, X
					}} \alpha \quad &
			  \label{eq:stochasticOCP_obj}\\
				\text{s.t. } \, \forall 	W_{[0,N-1]} \in \mcl{W}  , &\quad \forall k \in \set I_{[0,N-1]}, \notag \\
			\textstyle{\sum_{k=0}^{N-1}}	 \mean \big[ \|Y_{k}\|^2_Q + \|U_{k}\|^2_R\big] \leq \alpha,& \label{eq:stochasticOCP_obj_1}\\
				X_{k+1}=   AX_{k}+BU_{k} 
				+E W_{k},
			 &	\quad X_0  = x_\ini, \label{eq:OCPdynamic}
				\\
				Y_{k}=   CX_{k}+DU_{k} +F W_k,   &
				\label{eq:OCPdynamic_Y} \\
	U_{k} =  \bar{u}_k + \textstyle{\sum_{i=0}^{k-1}} K_{k,i} W_i,  &\quad U_0 =  \bar{u}_0,   
		 \quad  \label{eq:causal}\\
				\prob [a_{u,i}^\top U_k \leq 1 ]\geq 1 - \varepsilon_u,   &\quad \forall i \in \I_{[1,N_u]}, \label{eq:chance_U} \\
						\prob [a_{y,i}^\top Y_k  \leq  1 ]\geq 1 - \varepsilon_y,  &\quad \forall i \in \I_{[1,N_y]}    \label{eq:chance_Y}.
			\end{align}
	\end{subequations}
	Given the uncertainy set $\mcl W  \subset  (\splx{\dimw})^N$,  we minimize the worst-case value $\alpha \in \R$ of the objective function over the horizon $N \in \mbb{N}$ in  \eqref{eq:stochasticOCP_obj}--\eqref{eq:stochasticOCP_obj_1}.
The objective function  is  the expected value of a quadratic form with $Q \in \mbb S_+^{\dimy}$ and $R \in \mbb S_{++}^{\dimu}$. 
We consider  i.i.d.  disturbances directly entering the dynamics in \eqref{eq:OCPdynamic}-\eqref{eq:OCPdynamic_Y}.
Similar to \cite{Coppens2021,Lian2021} we aim at affine and causal disturbance feedback. This is encoded in \eqref{eq:causal} and it can be written as
\begin{equation}\label{eq:OCPcausal} 
\begin{gathered}
	U_{[0,N-1]} =  \bar{u}_{[0,N-1]} + K_w W_{[0,N-1]},  \\
	K_w  =   \left[\begin{smallmatrix}
			\mbf{0} &\mbf{0} &\cdots& \mbf{0}\\
			K_{1,0} &\mbf{0} & \cdots &   \mbf{0}\\
			\vdots & \ddots&\mbf{0} & \vdots \\
			K_{N-1,0}& \cdots  &	K_{N-1,N-2}& \mbf{0}\\
		\end{smallmatrix}\right] \in \R^{N\dimu\times N\dimw}. 
\end{gathered}
\end{equation}
 	Chance constraints   are specified as individual half-space constraints  by $a_{u,i} \in \R^{\dimu}$, $ i \in \I_{[1,N_u]}$, and $a_{y,i} \in \R^{\dimy}$,  $ i \in \I_{[1,N_y]}$ with  probabilities of $1 - \varepsilon_u$ and $1 - \varepsilon_y$, respectively,  in \eqref{eq:chance_U}-\eqref{eq:chance_Y}.

We remark that the conceptual formulation \eqref{eq:stochasticOCP} poses several challenges.
First, the optimization involves infinite-dimensional $\mcl L^2 $ random variables. Second, distributional robustness  requires  \eqref{eq:stochasticOCP_obj_1}--\eqref{eq:chance_Y}  to be satisfied for all  possible random variable sequences in $\mcl W$, resulting in infinitely many infinite-dimensional  constraints. To address these challenges, we  use the PCE framework to reformulate the random variables, the ambiguity sets, and the chance constraints.

\section{The PCE Perspective on Gelbrich Ambiguity}\label{sec:PCE}
\subsection{Primer on polynomial chaos expansion} \label{sec:PCEintro}
The core idea of PCE is that $\mcl{L}^2$ random variables can be expressed as a series expansion in a suitable basis \cite{sullivan15introduction}. To this end, consider an orthogonal polynomial basis $\{\phi^j(\xi)\}_{j=0}^{\infty}$ which spans $\spl$, i.e.
$
\basis[i]{j} = \int_{-\infty}^{\infty} \phi^i\left(\xi(\omega)\right) \phi^j\left(\xi(\omega)\right) \diff \mu(\omega) = \delta^{ij}\basis{j}
$
where $\delta^{ij}$ is the Kronecker delta. 
We remark that it is customary in PCE to consider $\phi^0=1$.

\begin{definition}[Polynomial chaos expansion]
	The PCE of a  random variable $V\in \spl$ with respect to the basis $\{\phi^j\}_{j=0}^{\infty}$ is $V = \sum_{j=0}^{\infty}\pce[j]{v} \phi^j$ with $\pce[j]{v} = \langle V, \phi^j \rangle/\basis{j}$, where $\pce[j]{v}$ is called the $j$-th PCE coefficient.\End
\end{definition}
We remark that by applying PCE component-wise the $j$th PCE coefficient vector of a vector-valued  random variable $V\in\splx{\dimv}$ reads
$
	\pce{v}^j = \begin{bmatrix} \pcecoe{v}{1,j} & \pcecoe{v}{2,j} & \cdots & \pcecoe{v}{\dimv,j} \end{bmatrix}^\top,
$
where $\pcecoe{v}{i,j}$ is the $j$th PCE coefficient of component $V^i$. Moreover, we introduce a shorthand of the matrix generated by horizontally stacking the PCE coefficients as 
$
\pcecoe{V}{[0,L-1]} \doteq \left[\pcecoe{v}{0},\pcecoe{v}{1},\dots,\pcecoe{v}{L-1}\right] \in \R^{\dimv \times L}.
$ 

\begin{definition}[Exact PCE representation \cite{muehlpfordt18comments}] \label{def:exactPCE}
	The PCE of a random variable $V \in \splx{n_v}$ is said to be exact with dimension $L$ if $ V -  \sum_{j=0}^{L-1}\pce[j]{v} \phi^j=0$. \End
\end{definition}

Furthermore, with Definition~\ref{def:exactPCE}, the expectation $\mean[V]$, the variance $\mbb{V}[V]$, and the covariance $\Sigma [V,\tilde{V}]$ can be obtained from the PCE coefficients as $\mean [V] = \pce{v}^0$, $\mbb V [V] =\sum_{j=1}^{L-1} \pce{v}^j\circ \pce{v}^j\basis{j}$,
$\Sigma[V,\tilde{V}] = \sum_{j=1}^{L-1} \pce{v}^j\tilde{\pce{v}}^{j\top}\basis{j}$,
where  
$\pce{v}^j\circ \pce{v}^j$ refers to the Hadamard product~\cite{lefebvre20moment}.

\subsection{PCE representation of disturbances}

	For i.i.d. (not necessarily Gaussian) disturbances $W_k$, $k\in\N$, we first construct an exact PCE of finite dimension. For starters, we denote the  map $\Psi: \mbb{S}_{+}^\dimw \to \R^{\dimw\times \dimw},$ $\wcov \mapsto \Psi(\wcov)$ as a \textit{generalized matrix square root} if it is bijective and satisfies $	\wcov = \Psi(\wcov)\Psi(\wcov)^\top
	$.

	Consider $\xi_k$ with $\mean[\xi_k]=0$ and $\Sigma[\xi_k,\xi_k] = I_{\dimw}$ such that 	$
	W_k= \wmean  + \Psi(\wcov) \xi_k$  holds.
	Notice that the elements of $\xi_k$---i.e. $\xi^i_k$, $i\in \I_{[1,\dimw]}$---are independently distributed and satisfy  $\mean[\xi^i_k]=0$ as well as $\mbb V[\xi^i_k]=1$.
	Using the basis  $\{\phi^j_w(\xi_k)\}_{j=0}^{\dimw} = \{1,\{\xi^i_k\}_{i=1}^{\dimw}\}$ with polynomials of degree of at most $1$, the exact and finite PCE of $W_k$ is obtained as
		\begin{align}
			W_k=\wmean  + \Psi(\wcov) \xi_k=\textstyle \sum_{j=0}^{\dimw} \pcecoe{w}{j} \phi^j_w(\xi_k),  \label{eq:PCEGaussian}
		\end{align}
		with $\pcecoe{w}{0} = \wmean$ and $\pcecoe{W}{[1,\dimw]} \doteq \left[\pcecoe{w}{1},\dots,\pcecoe{w}{n_w}\right] = \Psi(\wcov)$.

For any finite horizon $N \in \N$  in OCP \eqref{eq:stochasticOCP}  and let the inputs $U_k$ satisfy \eqref{eq:causal} the following orthonormal basis 
\begin{equation}\label{eq:common_basis}
	\{\phi^j(\boldsymbol \xi)\}_{j=0}^{L-1} = \{1,\{\{\xi^i_k\}_{i=1}^{\dimw}\}_{k=0}^{N-1}\},
\end{equation}
where $\boldsymbol \xi = [\xi_{0}^\top,...,\xi_{N-1}^\top]^\top \in \splx{N\dimw}$ and $L = N\dimw+1$,
allows exact PCEs for $(U,Y,W,X)_k$, $k \in \I_{[0,N-1]}$, cf. \cite{Pan2022a}. 

Applying Galerkin projection onto the basis in~\eqref{eq:common_basis}  yields the dynamics of the PCE coefficients 
\begin{subequations}\label{eq:PCEdynamics}
	\begin{align}
		\pcecoe{x}{j}\inst{k+1} &= A\pcecoe{x}{j}\inst{k} +B\pcecoe{u}{j}\inst{k}+ E\pcecoe{w}{j}\inst{k},\quad \pcecoe{x}{j}_0=\delta^{0j} x_\ini, \\
		\pcecoe{y}{j}\inst{k} &= C\pcecoe{x}{j}\inst{k} + D\pcecoe{u}{j}\inst{k} +F  \pcecoe{w}{j}\inst{k},\quad \forall j \in \I_{[0,L-1]}
	\end{align}
\end{subequations}
where $\delta^{0j}$ is the Kronecker delta \cite{GhanSpan03}. 	Due to the i.i.d. property of $W_k$, the PCE coefficients for $W_{[0,N-1]}$ satisfy 
\begin{equation}\label{eq:Wcoeff}
	\left[\boldsymbol{1}_N \otimes  \pce{w}^{0}, I_N \otimes  \pce{W}^{[1,\dimw]}\right] = \pce{W}^{[0,L-1]}_{[0,N-1]}\,
\end{equation}
where  $\pce{W}^{[0,L-1]}_{[0,N-1]} \in \R^{ N \dimw\times L} $ is the vertically stacked block matrix comprising $\{\pce{W}^{[0,L-1]}_{k}\}_{k=0}^{N-1}$. 

	At first glance,  the PCE representation of $W_k$ in \eqref{eq:PCEGaussian} seemingly resembles a  usual moment-based representation.  However, using the generalized square root  of the covariance, we obtain a linear parametrization of $W_K$, which in turn simplifies the data-driven uncertainty propagation.  Furthermore, for all $W_k$ collecting  the normalized random variables $\xi_k$, $k \in \I_{[0,N-1]}$ in the basis~\eqref{eq:common_basis}, we obtain the coefficient dynamics \eqref{eq:PCEdynamics}. These dynamics are structurally similar to the original dynamics in random variables \eqref{eq:RVdynamics}. Put differently, for all $j \in \I_{[0,L-1]}$ the coefficient dynamics \eqref{eq:PCEdynamics} capture the influence of the corresponding disturbance component. We remark that considering the explicit state covariance propagation $\Sigma_{k+1} = (A+BK)\Sigma_k(A+BK)^\top +E\Sigma[W_k, W_k]E^\top$ would render it more difficult to work with data-driven system descriptions. We refer to \cite{Faulwasser2022} for a more detailed comparison of moment propagation and PCE.

\subsection{Representation of  Gelbrich ambiguity sets}\label{sec:reformulate}

The PCE reformulation of $W_k$ in \eqref{eq:PCEGaussian} suggest to translate  the  Gelbrich ambiguity set $\mcl{A}$ to an uncertainty set of the PCE coefficients, i.e., translation to a set of matrices with real numbers. Specifically, the distributions in $\mcl{A}$  are bijectively paired to  the PCE coefficients matrices by the  map  
\begin{subequations}\label{eq:Pi}
\begin{equation}
	\Pi_\Psi:~ \mu_W \in \mcl D(\wmean,\wcov) \mapsto \begin{bmatrix}
		\wmean \,|\, \Psi(\wcov)\end{bmatrix}.
\end{equation} 
Notice the design degree of freedom to use any generalized matrix square root $\Psi$.
As the principal square root $(\cdot)^{\frac{1}{2}}$ in the Gelbrich metric (Def.~\ref{def:Gelbrich}) is a non-convex function, we choose 
\begin{equation}
\Psi(\wcov) = \bar{\wcov}^{-\frac{1}{2}} ( \bar{\wcov}^{\frac{1}{2}}  \wcov \bar{\wcov}^{\frac{1}{2}} )^{\frac{1}{2}}.
\end{equation}
\end{subequations}
The map $\wcov \mapsto \Psi(\wcov)$ is bijective and it satisfies $\Psi(\wcov)\Psi(\wcov)^\top = \Gamma$.
For $(\bar{\wcov})^{-\frac{1}{2}}$ to exist, 
we assume $\bar{\wcov} \in \mbb S_{++}^{n_w}$. 
Moreover, consider the PCE coefficient ambiguity set
\begin{equation}\label{eq:Delta_WS}
	\mbb A	 =\left\{ \begin{gathered}\pcecoe{W}{[0,\dimw]} \in\\ \R^{\dimw\times (\dimw+1)} 	\end{gathered} \, \middle| \, \begin{gathered}
			\left\| \pcecoe{W}{[0,\dimw]} - \left[\bar{m} \, \middle|\,  \bar{\wcov}^{\frac{1}{2}}\right]\right\|\leq \rho\\
			\bar{\wcov}^{\frac{1}{2}}\pcecoe{W}{[1,\dimw]} \succeq 0
		\end{gathered}
		\right\}.
	\end{equation}  
\begin{lemma}[$\Pi_\Psi(\mcl A) = \mbb A$]\label{lem:bijective_WS}
	 Given the empirical moments $(\bar{m},\bar{\wcov})$  with mean $\bar{m} \in \R^{\dimw}$ and  covariance $\bar{\wcov}\in \mbb S_{++}^{\dimw}$. Consider 
	  $\Pi_\Psi$ from \eqref{eq:Pi}, 
	 the Gelbrich ambiguity set $\mcl{A}$ from \eqref{eq:ambiguity_WS}, and the PCE coefficient ambiguity set $\mbb A$ from \eqref{eq:Delta_WS}.
	Then, the element-wise image of $\mcl A$ under $\Pi_\Psi$ is given by $\mbb A$.
\end{lemma}
\begin{proof}
First, we show that under the map $\Pi_\Psi$, the Gelbrich distance $\mbb G$ in the definition of $\mcl A$~\eqref{eq:ambiguity_WS} corresponds to the norm expression in $\mbb A$~\eqref{eq:Delta_WS}.
With  $\Pi_\Psi$ and $\Psi$ as specified in~\eqref{eq:Pi}, we have  	$\pcecoe{w}{0} = m$, $\pcecoe{W}{[1,\dimw]} = \Psi(\wcov)$ and  $d= \mbb G\left( (m,\Gamma),(\bar m, \bar \Gamma) \right) $ 
$	= \sqrt{\|\pcecoe{w}{0} - \bar{m}\|^2 + \trace( \bar{\wcov} +\wcov -2 \bar{\wcov}^{\frac{1}{2}}\pcecoe{W}{[1,\dimw]}  )}.$
	Moreover, with $M = \bar{\wcov}^{\frac{1}{2}}-\pcecoe{W}{[1,\dimw]}  $ and since
	$\wcov = \pcecoe{W}{[1,\dimw]} \pcecoe{W}{[1,\dimw]\top}$, 
	 we have
$
		d = \sqrt{\|\pcecoe{w}{0} - \bar{m}\|^2 + \trace\left( M M^\top\right)}  =	\| \pcecoe{W}{[0,\dimw]} - [\bar{m} \, |\, \bar{\wcov}^{\frac{1}{2}}]\|,
$
	where we used the properties of the Frobenius norm. 
	
	Next we prove that $	\bar{\wcov}^{\frac{1}{2}}\pcecoe{W}{[1,\dimw]}\succeq 0$  in \eqref{eq:Delta_WS} is equivalent to $\Gamma \succeq 0$ in  \eqref{eq:ambiguity_WS} provided $\pcecoe{W}{[1,\dimw]} = \Psi(\wcov)$ as in~\eqref{eq:Pi}. That is, we aim to show 
		\begin{equation}\label{eq:equivalence}
		\pcecoe{W}{[1,\dimw]} = \Psi(\wcov), \Gamma \succeq 0 \Leftrightarrow \bar{\wcov}^{\frac{1}{2}}\pcecoe{W}{[1,\dimw]}\succeq 0.
		\end{equation}
  
  	 The $\Rightarrow$ implication holds, since $ \bar{\wcov}^{\frac{1}{2}}\pcecoe{W}{[1,\dimw]} =  ( \bar{\wcov}^{\frac{1}{2}}  \wcov \bar{\wcov}^{\frac{1}{2}} )^{\frac{1}{2}} \succeq 0$. 
  	$\Leftarrow$: since $\Psi$ is bijective, its inverse map $\Psi^{-1}: \R^{\dimw\times \dimw}\to \mbb{S}_{+}^\dimw  ,$ $ \pcecoe{W}{[1,\dimw]}\mapsto \Psi^{-1}(\pcecoe{W}{[1,\dimw]})$ exists. Thus, if the right hand side of \eqref{eq:equivalence} holds, we find  $\Gamma =  \Psi^{-1}(\pcecoe{W}{[1,\dimw]}) \in \mbb{S}_{+}^\dimw$ and then the left hand side holds. 
\end{proof}

	Recall that the Gelbrich distance in  Definition \ref{def:Gelbrich} is  a non-convex function of $(m,\Gamma)$. However,  it is convex in the PCE coefficients $\pcecoe{W}{[0,\dimw]}$. Hence the PCE ambiguity set $\mbb A$ from \eqref{eq:Delta_WS} is a compact and convex subset of $\R^{\dimw\times (\dimw+1)}$.
	Finally, we arrive at the uncertainty description for the PCE coefficient sequences $W_{[0,N-1]}$ 
	\begin{equation}\label{eq:uncertainty_PCE}
	\mbb W\doteq \left\{
		\begin{gathered}\pce{W}^{[0,L-1]}_{[0,N-1]} \in\\ \R^{N\dimw \times L} 	\end{gathered}	\, \middle|\,
		\begin{gathered}
			 \pce{W}^{[0,L-1]}_{[0,N-1]}\left(	\pce{W}^{[0,\dimw]}\right)\, \text{s.t. }\eqref{eq:Wcoeff}  \\
			\pce{W}^{[0,\dimw]}\in \mbb A
		\end{gathered}  \right\}, 
	\end{equation}
	and at the PCE reformulation of $\mcl W$ from \eqref{eq:uncertainty_RV}
	\begin{equation}\label{eq:pce_W}
	\mcl W =  \left\{ \begin{gathered} W	\end{gathered}	\, \middle|\,
	\begin{gathered}
	W = \textstyle{\sum_{j=0}^{L-1}}	\pce{W}^{j} \phi^j(\boldsymbol \xi), \, \phi \text{ cf. }\eqref{eq:common_basis}  \\
	\pce{W}^{[0,L-1]} \in \mbb W ,\,
	\boldsymbol{\xi}\in \mcl{D}(0,I_{N \dimw })
	\end{gathered}  \right\}.
	\end{equation}
	Figure~\ref{fig:tranlationofSets} summarizes the relations and  maps between  the ambiguity sets $\mcl A, \mbb A$ and the sequence uncertainty descriptions $\mcl W, \mbb W$.
	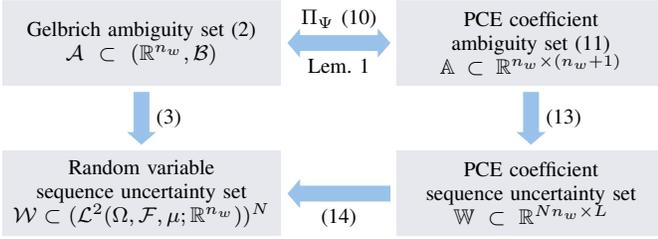
\begin{figure}[t]
		\centering
		\begin{tikzpicture}[my triangle/.style={-{Triangle[width=\the\dimexpr1.8\pgflinewidth,length=\the\dimexpr0.8\pgflinewidth]}}, my double triangle/.style={{Triangle[width=\the\dimexpr1.8\pgflinewidth,length=\the\dimexpr0.8\pgflinewidth]}-{Triangle[width=\the\dimexpr1.8\pgflinewidth,length=\the\dimexpr0.8\pgflinewidth]}}]
    \footnotesize
    \node[rectangle, draw=none, fill=gray1, minimum height=4em, text centered, text width=12.5em, outer sep=2.5pt] (A) {
    	Gelbrich ambiguity set  \eqref{eq:ambiguity_WS} \\$	\mcl A \subset (\R^{\dimw},\mcl{B})$ 
};
    \node[right=39pt of A, rectangle, draw=none, fill=gray1, minimum height=4em, text centered, text width=12em, outer sep=2.5pt] (B) {PCE coefficient ambiguity set  \eqref{eq:Delta_WS} \\ $\mbb A \subset \R^{\dimw\times (\dimw+1)}$};
    \node[below= 20pt of A,rectangle, draw=none, fill=gray1, minimum height=4em, text centered, text width=12.5em, outer sep=2.5pt] (C) { Random variable sequence uncertainty set\\ $\mcl W \subset (\splx{\dimw} )^N$};
    \node[below=20pt of B,rectangle, draw=none, fill=gray1, minimum height=4em, text centered, text width=12em, outer sep=2.5pt] (D) {PCE coefficient  sequence uncertainty set \\ $\mbb W \subset \R^{N\dimw \times L}$};
    
    \draw[line width=6pt,my double triangle, draw=lightblue](A) -- (B) node [above, midway] (n1) {$\Pi_\Psi$~\eqref{eq:Pi}};
    \draw[line width=6pt,my double triangle, draw=none](A) -- (B) node [below, midway] (n2) {Lem.~\ref{lem:bijective_WS}};
    \draw[line width=6pt,my triangle, draw=lightblue](A) -- (C) node [midway,right] (n3) {\eqref{eq:uncertainty_RV}};
    \draw[line width=6pt,my triangle, draw=lightblue](D) -- (C) node [midway, below] (n4) {\eqref{eq:pce_W}};
     \draw[line width=6pt,my triangle, draw=lightblue](B) -- (D) node [midway, right] (n5) {\eqref{eq:uncertainty_PCE}};
\end{tikzpicture}
		\caption{Relations and maps between the  sets $\mcl A, \mbb A$, $\mcl W$, and $\mbb W$.}
		\label{fig:tranlationofSets}
	\end{figure}
\section{Data-driven distributionally robust  optimal control in PCE cofficients}\label{sec:PCEOCP}
	The above reformulation of the ambiguity set $ \mcl A $ to $\mbb{W} $ enables us to the cast the distributionally robust OCP \eqref{eq:stochasticOCP} as an \textit{uncertain conic problem}, whereby  we will use a data-driven representation in lieu of  explicit knowledge of the system matrices.

\subsection{Data-driven representation of stochastic LTI systems}\label{sec:data_driven}
	For a specific uncertainty outcome $\omega \in \Omega$ the realization of $W_k$ is written as $w\inst{k} \doteq W\inst{k}\relx$, $k\in\N$. Likewise, the realizations of inputs, outputs, and states are 
	$u_k \doteq U\inst{k}\relx,$  $y_k \doteq Y\inst{k}\relx$, and $x_k \doteq X\inst{k}\relx$, respectively. Given  $\{w\inst{k}\}_{k\in \N}$, the stochastic system~\eqref{eq:RVdynamics} induces the realization dynamics
\begin{subequations}\label{eq:Realizationdynamics}
	\begin{align}
		x\inst{k+1} &= Ax\inst{k} +Bu\inst{k}+ Ew\inst{k},\quad x_0 =x_\ini \\
		y\inst{k} &= Cx\inst{k} + Du\inst{k} +F  w\inst{k}.
	\end{align}
\end{subequations}

\begin{assumption}[System properties and data] \label{ass:minimal_state_representation} 
	Consider stochastic LTI system \eqref{eq:RVdynamics} and its realization dynamics~\eqref{eq:Realizationdynamics}, we assume that
	$(A,[B,E])$ is a controllable pair, and respectively, $(A,C)$ is a observable pair. In addition, we suppose the matrices $(A, B, C, D, E, F)$ are unknown, while measurements of past input-output-disturbance realizations $u_k$, $y_k$ and $w_k$ are available. 
	\End 
\end{assumption}

\begin{definition}[Persistency of excitation \cite{willems05note}] Let $T, t \in \set{N}^+$. A sequence of inputs $\trar{u}{T-1}$ is said to be persistently exciting of order $t$ if the Hankel matrix
$
		\hankel_t(\trar{u}{T-1}) \doteq \left[\begin{smallmatrix}
			u\inst 0  &\cdots& u\inst{T-t} \\
			\vdots & \ddots & \vdots \\
			u\inst{t-1}& \cdots  & u\inst{T-1} \\
		\end{smallmatrix}\right]
$
	is of full row rank.  \End
\end{definition}

Next we recall crucial insights from \cite[Lem. 4, Cor. 2]{Pan2022a} which allow to represent the PCE coefficients dynamics~\eqref{eq:PCEdynamics} by previous recorded data of the realization dynamics \eqref{eq:Realizationdynamics}.
\begin{lemma}[\cite{Pan2022a}] \label{lem:stoch_fundamental}
	Let Assumption \ref{ass:minimal_state_representation} hold. Consider system \eqref{eq:RVdynamics} and a $T$-length realization trajectory tuple $\tra[0,T-1]{(u,w,y)}$ of its corresponding realization dynamics \eqref{eq:Realizationdynamics}.  We suppose that $\tra[0,T-1]{(u,w)}$ is persistently exciting of order $\dimx +t$.
	Then $ ( \tra{U}, \tra{W},\tra{Y})_{[0,t-1]}$ is a trajectory of \eqref{eq:RVdynamics} if and only if there exists $G \in \splx{T-t+1} $ such that 
	$\hankel_{t}(\tra[0,T-1]{v}) G=\tra[0,t-1]{V}$
	 holds for all 
	 $(\tra{v}, \tra{V})\in \{ (\tra{u},\tra{U}), (\tra{w}, \tra{W}),(\tra{y},\tra{Y})\} $. 
	Moreover, $ (\tra{\pce{u}}, \tra{\pce{w}}, \tra{\pce{y}})^j_{[0,t-1], }$, $j \in \I_{[0,L-1]}$ is a trajectory of the dynamics of PCE coefficients \eqref{eq:PCEdynamics}  if and only if there exists $\pce[j]{g} \in \R^{T-t+1}$ such that 
	$		\hankel_t(\tra[0,T-1]{v}) \pce[j]{g}=\tra{\pce{v}}_{[0,t-1]}^{j},~j\in \I_{ [0,L-1]},$
	holds for all
	 $(\tra{v}, \tra{\pce{v}})\in \{ (\tra{u},\tra{\pce{u}}), (\tra{w}, \tra{\pce{w}}), (\tra{y},\tra{\pce{y}})\}$.	
	\End
\end{lemma}
	It is worth be remarked that the structural similarity of the PCE coefficient dynamics~\eqref{eq:PCEdynamics} with \eqref{eq:RVdynamics} and \eqref{eq:Realizationdynamics} is at the core of the above lemma. Note that this similarity is jeopardized by co-variance based uncertainty propagation. 

\subsection{Distributionally robust data-driven OCP}\label{sec:data-drivenOCP}

Combining the above results, we turn to  the data-driven reformulation of OCP~\eqref{eq:stochasticOCP} in terms of PCE coefficients.
\begin{assumption}[Data availability]\label{ass:Data}
	Consider a given $T$-length realization trajectory tuple $\tra[0,T-1]{(u,w,y)}$ of the corresponding realization dynamics \eqref{eq:Realizationdynamics}.  We suppose that $\tra[0,T-1]{(u,w)}$ is persistently exciting of order $\dimx +N+T_\ini$ with $T_\ini$ not smaller than the system lag of \eqref{eq:RVdynamics}, cf. \cite{Pan2022a}. \End
\end{assumption}

	 Consider $T_\ini$ past measurements of $ (u,y,w)_{[-T_\ini,-1]}$ and a $T$-length realization trajectory tuple $\tra[0,T-1]{(u,w,y)}$ satisfying Assumption~\ref{ass:Data}.  Let  $\text{p}$ and $\text{f}$ denote the ranges $[-\Tini,-1]$ and $[0,N-1]$, respectively. Let $\mcl H_{v,\text{p}}$ and $\mcl H_{v,\text{f}}$ be the first $\Tini\dimv$ rows and, respectively, the remaining $N\dimv$ rows of the Hankel matrix  	$\hankel_{N+T_\ini}(\tra{v}_{[0,T-1]})$ for $\tra{v} \in \{\tra{u},\tra{y},\tra{w}\}$. 
	 Consider the stacked Hankel matrices as $\mcl H_{\text{p}}\doteq [\mcl H_{u,\text{p}}^\top,\mcl H_{y,\text{p}}^\top,\mcl H_{w,\text{p}}^\top]^\top$ and $\mcl H_{\text{f}}\doteq [\mcl H_{u,\text{f}}^\top,\mcl H_{y,\text{f}}^\top,\mcl H_{w,\text{f}}^\top]^\top$. The uncertainty set $\mbb W$ for the PCE coefficient sequences \eqref{eq:uncertainty_PCE} gives the  finite-dimensional and convex reformulation of OCP~\eqref{eq:stochasticOCP}
	\begin{subequations}\label{eq:PCEOCP} 
		\begin{align}
			\min_{ \substack{		\bar{u}, K, \alpha,
							\mathsf{u},\mathsf{y},\mathsf{g}
							}
				 } \alpha \quad &  \\
			\text{ s.t.} \quad  	     
			\forall \pcecoe{W}{[0,L-1]}_\text{f}  \in \mbb{W}, \,\forall k \in & \set I_{[0,N-1]},\nonumber\\
			\textstyle{\sum_{k=0}^{N-1} \sum_{j=0}^{L-1} }(\|\pcecoe{y}{j}_k\|_Q^2+ \|\pcecoe{u}{j}_k\|_R^2 ) & \leq \alpha, \label{eq:OCP_obj_PCE}  \\
			\mcl H_\text{p} \pcecoe{g}{j} =	\delta^{0j} [{u}^{\top}_{\text{p}},{y}^{\top}_{\text{p}},{w}^{\top}_{\text{p}}]^\top &,  \forall j \in \I_{[0,L-1]},
			\label{eq:OCP_Hankels_ini}\\
			\mcl H_{\text{f}}\pcecoe{g}{j} =
			[\pcecoe{u}{j\top}_{\text{f}},\pcecoe{y}{j\top}_{\text{f}},\pcecoe{w}{j\top}_{\text{f}}]^\top &, \forall j \in \I_{[0,L-1]}, \label{eq:OCP_Hankels} \\
			\pce{u}_{\text{f}}^{0} = \bar{u}_{\text{f}} + K_w \pce{w}_{\text{f}}^{0}, \,
			\pce{u}_{\text{f}}^{j}= K_w\pce{w}^{j}_{\text{f}}&, \forall j \in \I_{[1,L-1]}, \label{eq:causality_u}  \\
  			a_{u,i}^\top\pcecoe{u}{0}_k +	  \sigma(\varepsilon_u)
  			\| a_{u,i}^\top \pcecoe{U}{[1,L-1]}_k\|  \leq 1 &, \forall i \in \I_{[1,N_u]}, \label{eq:OCP_chanceu} \\
  			a_{y,i}^\top\pcecoe{y}{0}_k +	  \sigma(\varepsilon_y)
  			\| a_{y,i}^\top \pcecoe{Y}{[1,L-1]}_k\|  \leq 1 &, \forall i \in \I_{[1,N_y]}, \label{eq:OCP_chancey}
		\end{align}
	\end{subequations}
	where $\delta^{0j}$ is the Kronecker delta, $K_w$ collects all feedback gains $K_{k,i}$ similar to \eqref{eq:OCPcausal}, and $\sigma(\varepsilon) =  \sqrt{ (1-\varepsilon)/\varepsilon }$. 

Lemma~\ref{lem:stoch_fundamental} justifies the data-driven representation of the dynamics of PCE coefficients~\eqref{eq:PCEdynamics} in \eqref{eq:OCP_Hankels_ini}-\eqref{eq:OCP_Hankels}. 
Note that $\delta^{0j}$ in  \eqref{eq:OCP_Hankels_ini} specifies the  PCE coefficients of the initial condition to be zero  for $j>0$, i.e.,  we consider a deterministic initial condition.
Causality and affiness of polices in \eqref{eq:causal} is stated in \eqref{eq:causality_u}.
The next result gives the exactness of the reformulation of the chance constraints from  \eqref{eq:chance_U}-\eqref{eq:chance_Y}  to \eqref{eq:OCP_chanceu}-\eqref{eq:OCP_chancey}.
	\begin{proposition}[PCEs for DRO chance constraints]\label{pro:chance_reformulation}
		Consider a  random variable $V \in \splx{\dimv}$ with its PCE $V(\boldsymbol{\xi}) = \sum_{j=0}^{L-1} \pcecoe{v}{j} \phi^j(\boldsymbol{\xi})$ regarding the basis~\eqref{eq:common_basis}. For $a \in \R^{\dimv}$, the distributionally robust chance constraint
		$	\prob [a^\top V(\boldsymbol{\xi}) \leq 1 ] ]\geq 1 - \varepsilon$, $ \forall \boldsymbol{\xi} \in \mcl{D}(0,\I_{N\dimw})
		$
		is equivalent to
$	  a^\top\pcecoe{v}{0} +	  \sqrt{ (1-\varepsilon)/\varepsilon }
		\|a^\top\pcecoe{V}{[1,L-1]} \|  \leq 1.$
	\end{proposition}
\begin{proof} 
	 Using~\eqref{eq:common_basis} we have $V(\boldsymbol{\xi}) = \pcecoe{v}{0}  +  \pcecoe{V}{[1,L-1]} \boldsymbol{\xi} $, and thus the DRO chance constraint reads \[\prob \left[ a^\top\pcecoe{V}{[1,L-1]}\boldsymbol{\xi}  \leq  1 - a^\top\pcecoe{v}{0} \right]\geq 1 - \varepsilon,\, \forall \boldsymbol{\xi} \in \mcl{D}(0,\I_{N\dimw}).\]  Since this expression is bilinear in $\boldsymbol{\xi} \in $  $\splx{N\dimw}$ and the decision variables $\pcecoe{V}{[0,L-1]} \in \R^{\dimv \times L}$, it is equivalent to  $	 
	 a^\top\pcecoe{V}{[1,L-1]} \mean[\boldsymbol{\xi}]+	\sqrt{ (1-\varepsilon)/\varepsilon }  (a^\top\pcecoe{V}{[1,L-1]}  \Sigma[\boldsymbol{\xi},\boldsymbol{\xi}] $ $ \pcecoe{V}{[1,L-1]\top} a )^{1/2} \leq 1- a^\top\pcecoe{v}{0}$, cf. \cite[Th. 3.1]{calafiore06chance}. With $\mean[\boldsymbol{\xi}] = 0$ and  $\Sigma[\boldsymbol{\xi},\boldsymbol{\xi}] = I_{N\dimw}$, we conclude the assertion.
\end{proof}

\begin{theorem}[Equivalence of OCP minimizers]\label{thm:equivOCP} 
	Consider OCP \eqref{eq:stochasticOCP}  with the random-variable uncertainty set $\mcl W$ \eqref{eq:uncertainty_RV} 
	and OCP~\eqref{eq:PCEOCP} with the PCE uncertainty set $\mbb W$ \eqref{eq:uncertainty_PCE}.
	Let Assumptions~\ref{ass:minimal_state_representation}--\ref{ass:Data} and the conditions of Lemma~\ref{lem:bijective_WS} hold. Then, for any given initial condition $(u,y,w)_{[-T_\ini,-1]}$  for  OCP~\eqref{eq:PCEOCP}, there exists $x_\ini \in \R^{\dimx} $ for OCP~\eqref{eq:stochasticOCP} such that the sets of minimizers $(\bar{u}^\star,K^\star,\alpha^\star)$ of OCP~\eqref{eq:stochasticOCP} and OCP~\eqref{eq:PCEOCP} are the same.
\End
\end{theorem}
\begin{proof}
	The proof relies on that the PCE reformulation of all random variables in the basis \eqref{eq:common_basis} is exact and the omission of the basis from OCP \eqref{eq:stochasticOCP} to OCP~\eqref{eq:PCEOCP} is without loss of information. 
	Due to Assumption~\ref{ass:minimal_state_representation} the system is observable and the measurements of $(u,w,y)$ are exact. Hence $(u,y,w)_{[-T_\ini,-1]}$ determines a unique initial state $x_\ini$ in OCP~\eqref{eq:stochasticOCP} given $T_\ini$ is not smaller than the system lag.  
	
	Using the basis~\eqref{eq:common_basis} all random variables in OCP~\eqref{eq:stochasticOCP} admit exact PCEs with at most  $L = N\dimw+1$ terms cf. \cite[Prop.~1]{Pan2022a}. 
	Replacing all random variables with their PCEs the constraint  \eqref{eq:stochasticOCP_obj_1} is equivalent to \eqref{eq:OCP_obj_PCE}  due to the orthonormality of the basis~\eqref{eq:common_basis}. With Assumption~\ref{ass:Data},  \eqref{eq:OCP_Hankels_ini}-\eqref{eq:OCP_Hankels} exactly captures the PCE coefficient dynamics, cf. Lemma~\ref{lem:stoch_fundamental}. Moreover,  \eqref{eq:causality_u}  expresses the causal and affine policies \eqref{eq:causal} in PCE coefficients. 
	With \eqref{eq:pce_W} we  split  the uncertainty description $\forall W_\text{f} \in \mcl{W}$ into two conditions $\forall \pcecoe{W}{[0,L-1]}_\text{f}  \in \mbb{W}$ and $\forall \boldsymbol{\xi} \in \mcl{D}(0,\I_{N\dimw})$. Using the latter condition and applying Proposition~\ref{pro:chance_reformulation} the chance constraints \eqref{eq:chance_U}-\eqref{eq:chance_Y}  are exactly reformulated to \eqref{eq:OCP_chanceu}-\eqref{eq:OCP_chancey}. Notice that the reformulated objective constraint and chance constraints are independent of the PCE basis~\eqref{eq:common_basis}. 	
	Thus, without loss of information, we drop the basis and finally obtain OCP~\eqref{eq:PCEOCP}. Since the reformulation from  OCP \eqref{eq:stochasticOCP} to OCP~\eqref{eq:PCEOCP} is exact,  the sets of their minimizers restricted to the variables $(\bar{u}^\star,K^\star,\alpha^\star)$ coincide.
\end{proof}

\subsection{Numerical implementation}

Observe that OCP \eqref{eq:PCEOCP} is an uncertain conic problem. Hence tractable reformulations are possible for specific types of uncertainty sets  \cite{BenTal2009}. In our approach, we approximate the uncertainty set $\mbb{A} $ in \eqref{eq:Delta_WS} by a polytope and than define the approximation of $\mbb{W}$ accordingly.
To this end, we uniformly sample $s \in \N^+$ points $\delta^j \in \R^{\dimw\times (\dimw+1)}  $  from $\mbb A$ for $j\in \I_{[1,s]}$. \footnote{An intuitive strategy is to sample uniformly over hypercubes which contain $\mathbb A$ and to neglect any samples which are not in $\mathbb A$. }  
We  approximate  $\mbb{A} $ by the convex hull of the sample points, denoted as
$
\widetilde{\mbb{A}} =\textrm{Conv}(\delta^1,...,\delta^{s}) \subset \mbb A.
$
By linearly lifting each vertex of  $\widetilde{\mbb{A}}$ via \eqref{eq:Wcoeff}, we obtain $\widetilde{\mbb{W}}$ similarly as in \eqref{eq:uncertainty_PCE}.

We denote the vertices of $\widetilde{\mbb{W}} $ by $\Delta \doteq \left\{\widetilde{\delta}^j, j \in \{1, \dots,\tilde{s}\}\right\}$  which are a subset of the lifted  sample points with $\tilde{s} \leq s $. Replacing  $\mbb W$ with the countable set $\Delta$, we obtain
\begin{equation}\label{eq:approximated_OCP} 
		\min_{ \substack{		\bar{u}, K, \alpha,
				\mathsf{u},\mathsf{y},\mathsf{g}
			}
		} \alpha \quad 
		\text{s.t. }   	     
		\forall \pcecoe{W}{[0,L-1]}_\text{f}  \in \Delta,\, \eqref{eq:OCP_obj_PCE}-\eqref{eq:OCP_chancey} . 
\end{equation}
  Observe that with  \eqref{eq:OCP_chanceu}--\eqref{eq:OCP_chancey}, \eqref{eq:approximated_OCP}  is a second-order cone program whose computational complexity is $O(\sqrt{N(N_u+N_y)\tilde{s}})$, cf.~\cite{Lobo1998}. 
 Due to the tight page limit, a detailed analysis of the sample efficiency of the proposed approximation strategy is postponed to
future work. Instead, we demonstrate its efficacy numerically.

\section{Numerical Example}\label{sec:simulation}
We consider the discrete-time stochastic double integrator 
	\begin{align*}
	X_{k+1}&= \begin{bmatrix}
	1 & 1\\
	0& 1 
\end{bmatrix} X_k +\begin{bmatrix}
	0.5 \\
	1
\end{bmatrix} U_k + W_k,\,
Y_k &= \begin{bmatrix}
1 & 0
\end{bmatrix} X_k,
	\end{align*}
where the $W_k$ are i.i.d. with Gaussian mixture distributions. Especially, $\mu_W$ is the mixture of  $\mcl{N}( \left[\begin{smallmatrix}
		0.1  \\
		0.1 
	\end{smallmatrix}\right],0.01 I_2 )$ and  $\mcl{N}( \left[\begin{smallmatrix}
		-0.1  \\
		-0.1 
	\end{smallmatrix}\right],0.01 I_2 )$ with mean $m_{\text{true}} = [0,0]^\top$ and covariance $\Gamma_{\text{true}}= \left[\begin{smallmatrix}
		0.03 &  0.02 \\
		0.02 & 0.03
	\end{smallmatrix}\right]$.
Notice that these true values of mean and covariance are unknown to the OCP. We specify $T_\ini = 2$ which corresponds to the system lag. 
The weighting matrices are  $Q = R =1$ for $Y$  and $U$, and the prediction horizon is $N=10$.
Chance constraints on the input require  $U_k \leq 0.5$ and $U_k \geq -0.5$ to be satisfied  individually with  probability of no less than $80\%$  for $k =0,...,9$. 

 To construct OCP~\eqref{eq:approximated_OCP}  based on measured data, we first apply $70$ random inputs to the system and record the output responses as well as the realized disturbances. Then we use this data to construct Hankel matrices and to estimate the  moments of $W$ as
 $ [\bar{m}\,|\,\bar{\Gamma}]=\left[
 \begin{smallmatrix}
 	0.0025 \\
 0.0025 
 \end{smallmatrix}
 \,\middle|\,\begin{smallmatrix}
  0.0211 & 0.0100\\
 	 0.0100 & 0.0157
 \end{smallmatrix}.\right] 
 $
Using $[\bar{m}\,|\,\bar{\Gamma}]$ as the empirical moment pair and setting the radius $\rho = \bar{\rho} \cdot \|[\bar{m}\,|\,\bar{\Gamma}]\|$ for a user-chosen $\bar{\rho}\in \R^+$, we obtain Gelbrich ambiguity sets $\mcl{A}$ \eqref{eq:ambiguity_WS} and the corresponding PCE uncertainty sets $\mbb A$ \eqref{eq:Delta_WS}.
To construct $\Delta$, we uniformly sample $s$ points from $\mbb A$. Subsequently, we investigate the effect of varying radius  $\rho$ and the number of samples $s$. 

We consider  three cases of OCP~\eqref{eq:approximated_OCP} : 
\begin{itemize}
\item[(I)]  The robust case, where OCP~\eqref{eq:approximated_OCP}  is solved with $\Delta$ for different values of $\bar{\rho}$ and $s$. 

\item[(II)]  The optimistic case, where OCP~\eqref{eq:approximated_OCP}  is solved with 
	$\Delta = \left\{
	[\bar{m},|,\Psi(\bar{\Gamma})]
	\right\}$, using the empirical moments estimated from the 70 recorded  disturbance samples. 

\item[(III)]  The ideal case, i.e.,  OCP~\eqref{eq:approximated_OCP}   with $\Delta =\left\{[m_\text{true} \,|\,\Psi({\Gamma_{\text{true}}})]\right\}$, utilizing the true moments.
\end{itemize}
Each OCP is solved using the same initial data $u_{\text{p}}$, $y_{\text{p}}$, and ${w}_{\text{p}}$. Note that with ambiguity sets of fixed moments, cases II and III are instances of the approach in \cite{Li2021}.

Using $1000$ different sampled disturbance realization sequence of length $10$ each, Table \ref{tab:compare} compares the averaged cost $J$ and the number of constraint violations $\#_V$ for case I with different values of $\bar{\rho}$ and $s$ with cases II \& III.
We see that increasing $\bar{\rho}$ and $s$ leads to fewer constraint violations and decreased performance.
Comparing case I with cases II \& III, it is evident that the former provides a more robust solution. 
 Figure~\ref{fig:response} shows the corresponding input and output responses 	of  case I with $\bar{\rho} = 0.5$ and $s = 50$ as well as cases II \& III.  Observe that the input responses of case I violate the constraints much less  frequently compared to case II (with moments estimated from data) and still achieve similar output responses as case III (with the true moments).

\begin{table}[t]
	\caption{Comparison of the averaged cost $ J$ and the number of constraints violation $\#_V$ for 1000 realization trajectories. }\label{tab:compare}
				
	\centering
	\begin{adjustbox}{width=\columnwidth}
		\centering
		\begin{tabular}{c ccccccc c  }	
			\toprule
		case I	&\multicolumn{2}{c}{$\bar{\rho} = 0.1$} & \multicolumn{2}{c}{ $\bar{\rho} = 0.3 $}   &  \multicolumn{2}{c}{ $\bar{\rho} = 0.5 $} & \multicolumn{2}{c}{ $\bar{\rho} = 0.7$}   \\
		s	&  $J$ 	& $\#_V$& $J$&$\#_V$  &$J$ & $\#_V$ & $J$ & $\#_V$ 	\\
			\midrule
			10 & 24.54	& 138 &24.79 &68	&	25.24	 & 27 &	 25.56&	11\\
			50  &24.58& 124& 25.15 &53 &25.79& 24 &26.65 & 6 
			\\ 
			100 & 24.58& 124&	25.15&53&25.79&24& 26.65&6
			\\ 
			\bottomrule
		\end{tabular}
	\end{adjustbox}
\\
\centering
		\begin{tabular}{c ccc cc}	
	&  $J$ 	& $\#_V$& &$J$&$\#_V$ \\
	case II & 24.47 & 184 & case III  & 25.04&26 \\	
	\bottomrule
\end{tabular}
\end{table}

\begin{figure}[tb]
	\centering
		\includegraphics[width=8.7cm]{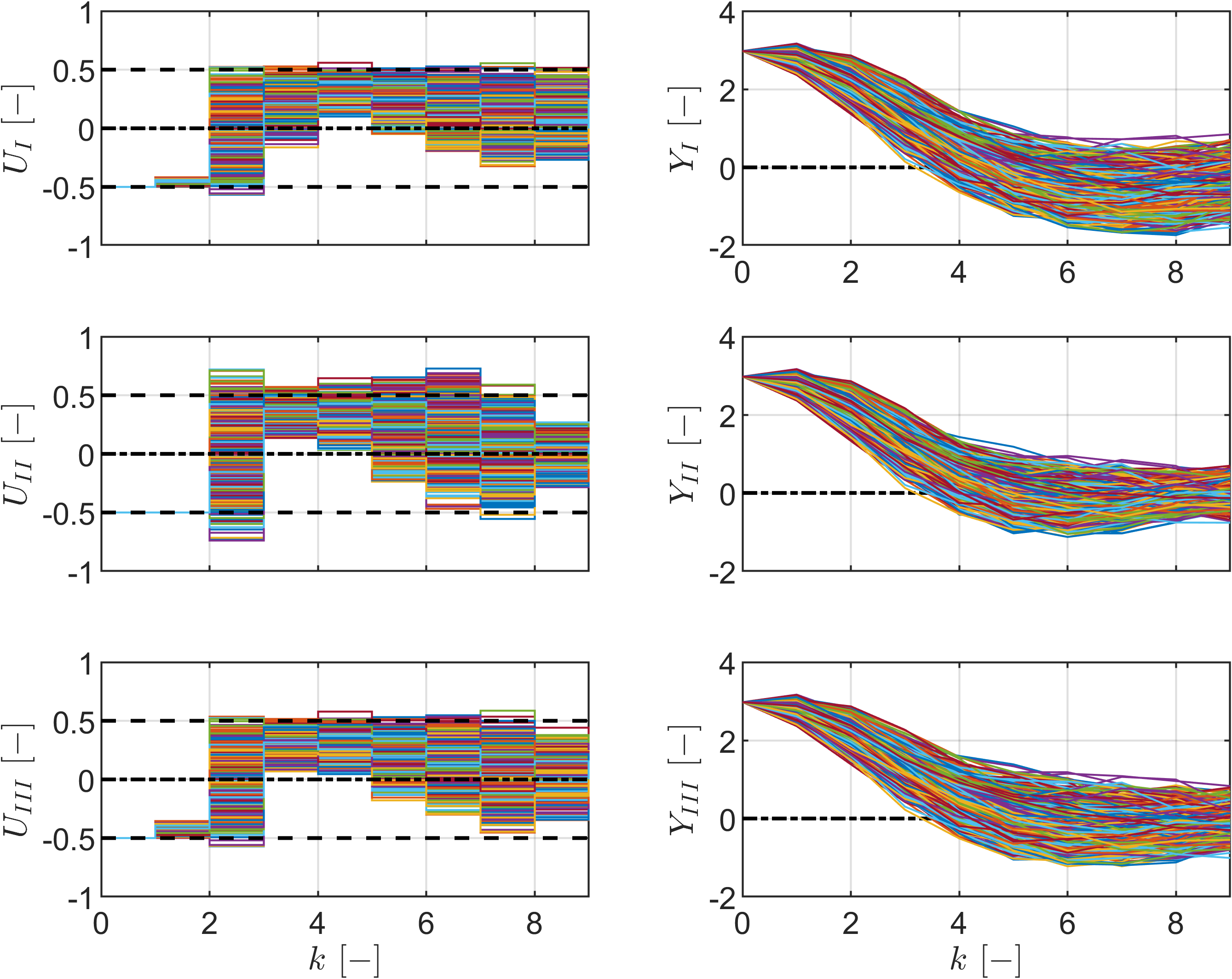}
		\caption{Input and output  response  for 1000 disturbance sequences of case I with $\bar{\rho} = 0.5$, $s=50$ (top), case II (middle), and case III (bottom). \label{fig:response}} 		
\end{figure}

\section{Conclusion and Outlook}\label{sec:conclusion}
This paper discussed distributionally robust uncertainty propagation for LTI systems via data-driven stochastic optimal control. We leveraged polynomial chaos expansions to derive an exact reformulation of model-based distributionally robust OCPs with Gelbrich ambiguity sets to data-driven uncertain conic problems with a finite-dimensional convex uncertainty set in PCE coefficients. A tractable approximation to convex programs has been proposed and illustrated via an example. 
Future work will consider tailored sampling strategies for the PCE coefficient ambiguity set, exact reformulations for  robust second-order cone constraints \cite{BenTal2009}, and  the effect of  the size of the previous recorded data.

\bibliographystyle{IEEEtran}
\bibliography{IEEEabrv}

\begin{thebibliography}{10}
\providecommand{\url}[1]{#1}
\csname url@samestyle\endcsname
\providecommand{\newblock}{\relax}
\providecommand{\bibinfo}[2]{#2}
\providecommand{\BIBentrySTDinterwordspacing}{\spaceskip=0pt\relax}
\providecommand{\BIBentryALTinterwordstretchfactor}{4}
\providecommand{\BIBentryALTinterwordspacing}{\spaceskip=\fontdimen2\font plus
\BIBentryALTinterwordstretchfactor\fontdimen3\font minus
  \fontdimen4\font\relax}
\providecommand{\BIBforeignlanguage}[2]{{%
\expandafter\ifx\csname l@#1\endcsname\relax
\typeout{** WARNING: IEEEtran.bst: No hyphenation pattern has been}%
\typeout{** loaded for the language `#1'. Using the pattern for}%
\typeout{** the default language instead.}%
\else
\language=\csname l@#1\endcsname
\fi
#2}}
\providecommand{\BIBdecl}{\relax}
\BIBdecl

\bibitem{Wiesemann2014}
W.~Wiesemann, D.~Kuhn, and M.~Sim, ``Distributionally robust convex
  optimization,'' \emph{Operations Research}, vol.~62, no.~6, pp. 1358--1376,
  2014.

\bibitem{Fochesato2022}
M.~Fochesato and J.~Lygeros, ``Data-driven distributionally robust bounds for
  stochastic model predictive control,'' in \emph{2022 IEEE 61st Conference on
  Decision and Control}.\hskip 1em plus 0.5em minus 0.4em\relax IEEE, 2022, pp.
  3611--3616.

\bibitem{Lu2020}
S.~Lu, J.~H. Lee, and F.~You, ``Soft-constrained model predictive control based
  on data-driven distributionally robust optimization,'' \emph{AIChE Journal},
  vol.~66, no.~10, p. e16546, 2020.

\bibitem{Coppens2021}
P.~Coppens and P.~Patrinos, ``Data-driven distributionally robust {MPC} for
  constrained stochastic systems,'' \emph{IEEE Control Systems Letters},
  vol.~6, pp. 1274--1279, 2021.

\bibitem{Coulson21a}
J.~Coulson, J.~Lygeros, and F.~Dörfler, ``Distributionally robust chance
  constrained data-enabled predictive control,'' \emph{IEEE Transactions on
  Automatic Control}, vol.~67, no.~7, pp. 3289--3304, 2022.

\bibitem{Pan2022a}
G.~Pan, R.~Ou, and T.~Faulwasser, ``On a stochastic fundamental lemma and its
  use for data-driven optimal control,'' \emph{IEEE Transactions on Automatic
  Control}, pp. 1--16, 2022.

\bibitem{Faulwasser2022}
T.~Faulwasser, R.~Ou, G.~Pan, P.~Schmitz, and K.~Worthmann, ``Behavioral theory
  for stochastic systems? {A} data-driven journey from {Willems to Wiener} and
  back again,'' \emph{Annual Reviews in Control}, vol.~55, pp. 92--117, 2023.

\bibitem{VanParys2015}
B.~P.~G. Van~Parys, D.~Kuhn, P.~J. Goulart, and M.~Morari, ``Distributionally
  robust control of constrained stochastic systems,'' \emph{IEEE Transactions
  on Automatic Control}, vol.~61, no.~2, pp. 430--442, 2015.

\bibitem{Li2021}
B.~Li, Y.~Tan, A.-G. Wu, and G.-R. Duan, ``A distributionally robust
  optimization based method for stochastic model predictive control,''
  \emph{IEEE Transactions on Automatic Control}, vol.~67, no.~11, pp.
  5762--5776, 2021.

\bibitem{sullivan15introduction}
T.~J. Sullivan, \emph{{Introduction to Uncertainty Quantification}}.\hskip 1em
  plus 0.5em minus 0.4em\relax Springer, 2015, vol.~63.

\bibitem{Givens1984}
C.~R. Givens and R.~M. Shortt, ``A class of {Wasserstein} metrics for
  probability distributions.'' \emph{Michigan Mathematical Journal}, vol.~31,
  no.~2, pp. 231--240, 1984.

\bibitem{Nguyen2021}
V.~A. Nguyen, S.~Shafieezadeh-Abadeh, D.~Kuhn, and P.~Mohajerin~Esfahani,
  ``{Bridging Bayesian and minimax mean square error estimation via Wasserstein
  distributionally robust optimization},'' \emph{Mathematics of Operations
  Research}, 2021.

\bibitem{Lian2021}
Y.~Lian and C.~N. Jones, ``From system level synthesis to robust closed-loop
  data-enabled predictive control,'' in \emph{2021 60th IEEE Conference on
  Decision and Control}.\hskip 1em plus 0.5em minus 0.4em\relax IEEE, 2021, pp.
  1478--1483.

\bibitem{muehlpfordt18comments}
T.~M{\"u}hlpfordt, R.~Findeisen, V.~Hagenmeyer, and T.~Faulwasser, ``Comments
  on quantifying truncation errors for polynomial chaos expansions,''
  \emph{IEEE Control Systems Letters}, vol.~2, no.~1, pp. 169--174, 2018.

\bibitem{lefebvre20moment}
T.~Lefebvre, ``On moment estimation from polynomial chaos expansion models,''
  \emph{IEEE Control Systems Letters}, vol.~5, no.~5, pp. 1519--1524, 2020.

\bibitem{GhanSpan03}
R.~G. Ghanem and P.~D. Spanos, \emph{{Stochastic Finite Elements: A Spectral
  Approach}}, revised~ed.\hskip 1em plus 0.5em minus 0.4em\relax Springer New
  York, 2003.

\bibitem{willems05note}
J.~C. Willems, P.~Rapisarda, I.~Markovsky, and B.~L.~M. De~Moor, ``A note on
  persistency of excitation,'' \emph{Systems \& Control Letters}, vol.~54,
  no.~4, pp. 325--329, 2005.

\bibitem{calafiore06chance}
G.~C. Calafiore and L.~E. Ghaoui, ``On distributionally robust
  chance-constrained linear programs,'' \emph{Journal of Optimization Theory
  and Applications}, vol. 130, no.~1, pp. 1--22, 2006.

\bibitem{BenTal2009}
A.~Ben-Tal, L.~El~Ghaoui, and A.~Nemirovski, \emph{Robust optimization}.\hskip
  1em plus 0.5em minus 0.4em\relax Princeton University Press, 2009.

\bibitem{Lobo1998}
M.~S. Lobo, L.~Vandenberghe, S.~Boyd, and H.~Lebret, ``Applications of
  second-order cone programming,'' \emph{Linear algebra and its applications},
  vol. 284, no. 1-3, pp. 193--228, 1998.

\end{thebibliography}

\end{document}